\documentclass[11pt,a4paper]{article}

\usepackage{graphicx}
\usepackage[english]{babel}
\usepackage{latexsym}
\usepackage{url}
\usepackage{amssymb}
\usepackage{amsmath}
\usepackage{amsthm}
\usepackage{amsfonts}
\usepackage{ifthen}

\usepackage[left=1.2in,top=1.2in,right=1.2in,bottom=1.2in,nohead]{geometry}

\newtheorem{theorem}{Theorem}
\newtheorem{corollary}[theorem]{Corollary}
\newtheorem{lemma}[theorem]{Lemma}

\newcommand{\RR}{\ensuremath{\mathbb R}}  
\newcommand{\NN}{\ensuremath{\mathbb N}}  
\newcommand{\D}{\mathcal{D}}  
\newcommand\CR{\hbox{\tt cr}}		  
\newcommand\MC{\hbox{\tt mwc}}		  

\newcommand\eps{\varepsilon}

\def\DEF#1{\textbf{\emph{#1}}}

\begin{document}

\title{Hardness of Approximation for Crossing Number\thanks{This work has been partially financed 
		by the Slovenian Research Agency, program P1-0297, project J1-4106, and
		within the EUROCORES Programme EUROGIGA (project GReGAS) of the European Science Foundation.
		Research partially carried out during the BIRS Workshop Crossing Numbers Turn Useful, 2011.}}

\author{Sergio Cabello\thanks{Department of Mathematics, IMFM, and 
				Department of Mathematics, FMF, University of Ljubljana, Slovenia.	
				email: {\tt sergio.cabello@fmf.uni-lj.si}}
}

\date{\today}

\maketitle

\begin{abstract}
We show that, if P$\not=$NP, there is a constant $c_0>1$ such that there is no $c_0$-approximation algorithm
for the crossing number, even when restricted to $3$-regular graphs.
\end{abstract}

\section{Introduction}

A \DEF{drawing} of a graph $G$ is a mapping $\D$ associating a point $\D(v)\in \RR^2$
to each vertex $v\in V(G)$ and a simple, polygonal path $\D(e)$ to each edge $e\in E(G)$ with the following
properties:
\begin{itemize}
\item for any two distinct vertices $u,v\in V(G)$, $\D(u)\not= \D(v)$;
\item for every edge $uv\in E(G)$, the endpoints of the path $\D(uv)$ are $\D(u)$ and $\D(v)$;
\item for every edge $e\in E(G)$ and every vertex $u\in V(G)$, the (relative) interior of the path $\D(e)$ is disjoint from $\D(u)$.
\end{itemize}

A \DEF{crossing} in a drawing $\D$ of a graph $G$ is a pair $(\{e,e'\},p)$, where $e$ and $e'$ are distinct edges of $G$
and $p\in \RR^2$ is a point that belongs to the interior of the paths $\D(e)$ and $\D(e')$.
The number of crossings of a drawing $\D$ is denoted by $\CR(\D)$ and is called the
\DEF{crossing number} of the drawing.
The \DEF{crossing number} $\CR(G)$ of a graph $G$ is the minimum $\CR(\D)$
taken over all drawings $\D$ of $G$. 
A drawing without crossings is an \DEF{embedding}.

Garey and Johnson~\cite{GJ} showed that the following optimization problem 
is NP-hard.
\begin{quote}
	{\sc CrossingNumber}.\\
	\emph{Instance:} A graph $G$.\\
	\emph{Feasible solutions:} Drawings of $G$.\\
	\emph{Measure:} Crossing number of the drawing.\\
	\emph{Goal:} Minimization.
\end{quote}
This result has been extended in several directions.
Hlin\v{e}n\'y~\cite{Hli} proved that the problem remains NP-hard for \DEF{cubic} graphs (3-regular graphs). 
This was reproved using crossing numbers with rotation systems by Pelsmajer et~al.~\cite{PSS}.
In a recent paper with Mohar~\cite{cm-socg-2010} we have shown that computing 
the crossing number for near-planar graphs is NP-hard. A graph is \DEF{near-planar} if
it is obtained from a planar graph by adding one edge.

None of the these proofs implies inapproximability of {\sc CrossingNumber} under the assumption that P$\not=$NP.
However, under stronger assumptions, inapproximability can be obtained from known results.
More precisely, the NP-hardness proof of Garey and Johnson~\cite{GJ} is from {\sc LinearArrangement}
and it implies that any inapproximability result for {\sc LinearArrangement} 
carries into an inapproximability result for {\sc CrossingNumber}. 
Amb{\"u}hl et al.~\cite{AMS} have recently shown that there
is no polynomial-time approximation scheme (PTAS) for {\sc LinearArrangement}, unless
NP-complete problems can be solved in randomized subexponential time.
(The precise assumption is that {\sc Satisfiability} cannot be solved 
in probabilistic time $2^{n^\epsilon}$ for any constant $\epsilon>0$.)
This directly implies that there is no PTAS for {\sc CrossingNumber}, unless 
NP-complete problems can be solved in randomized subexponential time. 
Although the NP-hardness proofs for cubic graphs by Hlin\v{e}n\'y~\cite{Hli} and Pelsmajer et~al.~\cite{PSS} also use 
reductions from {\sc LinearArrangement}, they do not imply any inapproximability 
because the value of the optimal linear arrangement is a lower-order term in the crossing number of the graphs constructed
in the reduction.

In this paper we show that, if P$\not=$NP, there is some constant $c_0>1$
such that there is no $c_0$-approximation for {\sc CrossingNumber}. The result holds also for cubic graphs.
Therefore, we strength the result mentioned in the previous paragraph by weakening the hypothesis.
Moreover, our reduction also implies inapproximability for cubic graphs, which was not known before under
any assumption.
We also provide a conceptually new proof of NP-hardness because we reduce from {\sc MultiwayCut}.
As noted by Hlin\v{e}n\'y~\cite{Hli}, for cubic graphs the minor crossing number is equal to the crossing number.
Thus, we also obtain inapproximability results for the minor crossing number.

On the positive side, the best approximation algorithm for {\sc CrossingNumber}, by Chuzhoy~\cite{chuzhoy},
has an approximation factor of $O(n^{9/10}\operatorname{poly}(\Delta \log n))$ for graphs with $n$ vertices
and maximum degree $\Delta$.  
It is worth noting that computing the crossing number is fixed-parameter tractable with respect
to the crossing number itself~\cite{G,KR}. Research on crossing number has been very active.
Vrt'o~\cite{Vrto} lists over 600 references.

\section{Preliminaries}
\label{sec:prelim}

\paragraph{Edge weights.} Our construction will be easier to describe if we work with weighted edges.  
The weights will always be positive integers.
Assume that $G$ is an edge-weighted graph where the weight of each edge $e$ is 
denoted by $w_e\in \NN$. The intuition is that $w_e$ tells how many parallel edges are represented by $e$. 
The crossing number of a drawing for such edge-weighted graph is defined by taking the sum of $w_e\cdot w_{e'}$, 
over all crossings $(\{ e,e'\},p)$ of the drawing. Again, the crossing number of such edge-weighted
graph is defined as the minimum of the crossing numbers over all drawings. 

Let $G$ be an edge-weighted graph.
We can construct an unweighted graph $\phi(G)$ from $G$ by replacing each edge $uv\in E(G)$ 
with a family $P_{uv}$ of $w_e$ parallel paths of length $2$ that connect $u$ to $v$.
It is easy to see that $\CR(G)=\CR(\phi(G))$. Indeed, any drawing $\D_G$ of $G$ gives rise to a drawing $\D_{\phi(G)}$ of $\phi(G)$ 
with $\CR(\D_{\phi(G)})=\CR(\D_G)$ by drawing each family $P_e$ within a small neighborhood of $\D_G(e)$.
On the other hand, any drawing $\D_{\phi(G)}$ of $\phi(G)$ can be used to construct a drawing $\D_G$ of $G$ with 
$\CR(\D_G)\le \CR(\D_{\phi(G)})$ by drawing each edge $e\in E(G)$ along the path of $P_e$ that participates in fewer crossings.

When the weights $w_e\in \NN$, $e\in E(G)$, are all bounded by a polynomial in $|V(G)|$, then the graph
$\phi(G)$ can be constructed from $G$ in polynomial time.

\paragraph{Rotation systems.} 
A \DEF{rotation system} in a graph $G$ is a list $\pi=(\pi_v)_{v\in V(G)}$, where each $\pi_v$ is a cyclic ordering
of the edges of $G$ incident to $v$. 
A drawing $\D$ of a graph $G$ agrees with the rotation system $\pi$ if, for each vertex $v\in V(G)$, the 
clockwise ordering around $\D(v)$ of the drawings of the edges incident to $v$ is the same as the cyclic
ordering $\pi_v$.
For a graph $G$ and a rotation system $\pi$ in $G$, we define $\CR(G,\pi)$ as the minimum of the crossing numbers 
over all drawings of $G$ that agree with $\pi$.
This concept can easily be extended to edge-weighted graphs. 

If $G$ is an edge-weighted graph and $\pi$ is a rotation system in $G$, 
we can define a rotation system $\phi_\pi(G,\pi)$ in $\phi(G)$:
for each edge $uv\in E(G)$, we replace in $\pi_u$ the edge $uv$ by the edges of $P_{uv}$ incident
to $u$ in such a way that the cyclic ordering of the paths in $P_{uv}$ are opposite at $u$ and $v$.
This implies that the paths of $P_{uv}$ can be drawn without crossings among themselves.
The same argument that was used above shows that 
$\CR(G,\pi)=\CR(\phi(G),\phi_\pi(G,\pi))$. The rotation $\phi_\pi(G,\pi)$ can be computed in polynomial time
provided that the edge-weights of $G$ are bounded by a polynomial in $|V(G)|$. 

A \DEF{combinatorial embedding} of a graph $G$ is a rotation system $\pi$ such that some embedding $\D$ of
$G$ agrees with $\pi$. Whitney's theorem states that a $3$-connected, planar graph 
has a unique combinatorial embedding~\cite[Chapter 4]{Diestel}. 

Consider any graph $G$ and any rotation system $\pi$ in $G$. In an optimal drawing of $G$ 
that agrees with $\pi$, each pair of edges participates in at most one crossing.
Indeed, if the edges $e$ and $e'$ would participate in two crossings $(\{e,e'\},p)$ and $(\{e,e'\},p')$,
we could obtain another drawing with fewer crossings: we exchange the portions of $D(e)$ and $D(e')$ between $p$ and $p'$ 
and then perturb the drawing around $p$ and $p'$ slightly to avoid the intersection.
In particular, for any rotation system $\pi$ of the complete graph $K_t$ we have $\CR(K_t,\pi)\le t^4$.

\paragraph{Multiway cut.} 
Our reduction will be from the following optimization problem about connectivity:
\begin{quote}
	{\sc MultiwayCut}.\\
	\emph{Instance:} A pair $(G,T)$ where $G$ is a connected graph and $T\subset V(G)$.\\
	\emph{Feasible solutions:} Sets of edges $F\subseteq E(G)$ such that, for each distinct $t,t'\in T$,
		there is no path in $G-F$ connecting $t$ to $t'$.\\
	\emph{Measure:} Cardinality of $F$.\\
	\emph{Goal:} Minimization.
\end{quote}

The set $T$ is the set of \DEF{terminals}.
Dahlhaus et al.~\cite{djpsy-94} proved 
that {\sc MultiwayCut} is MAX SNP-hard even when restricted
to instances with $3$ terminals\footnote{In Dahlhaus et al.~\cite{djpsy-94}
the problem was called multiterminal cut, but most recent works refer to it as multiway cut.}. 
This implies that there is a constant $c_M>1$ such 
that there is no $c_M$-approximation algorithm for {\sc MultiwayCut} with $3$ terminals, unless P$\not=$NP.
(In particular the problem is APX-hard for $3$ terminals; see~\cite{aetal-book}.)
We will only use instances $(G,T)$ of {\sc MultiwayCut} with $|T|=3$.
We denote by $\MC(G,T)$ the size of an optimal solution for $(G,T)$.

\paragraph{Notation.} We use $[3]=\{1,2,3\}$ and, for the rest of the paper, 
the indices depending on $i$ are always taken modulo $3$.

\section{From Multiway Cut to Crossing Number}
\label{sec:hard1}

Let $A$ be the graph defined by
\begin{align*}
		V(A)~&=~ \{ a_1,b_1,x_1,a_2,b_2,x_2,a_3,b_3,x_3,a_4\}, \\
		E(A)~&=~ \bigcup_{i\in [3]} \{ a_ib_i, a_ia_4, b_i b_{i+1}, a_ix_{i+1}, a_ix_{i-1}\}.
\end{align*}
The graph $A$ is shown in Figure~\ref{fig:graphD}, where it is clear that $A$ is planar.
Furthermore, it is a subdivision of a $3$-connected graph, and thus it has
a unique combinatorial embedding. 

\begin{figure}
\centering
	\includegraphics[width=.5\textwidth]{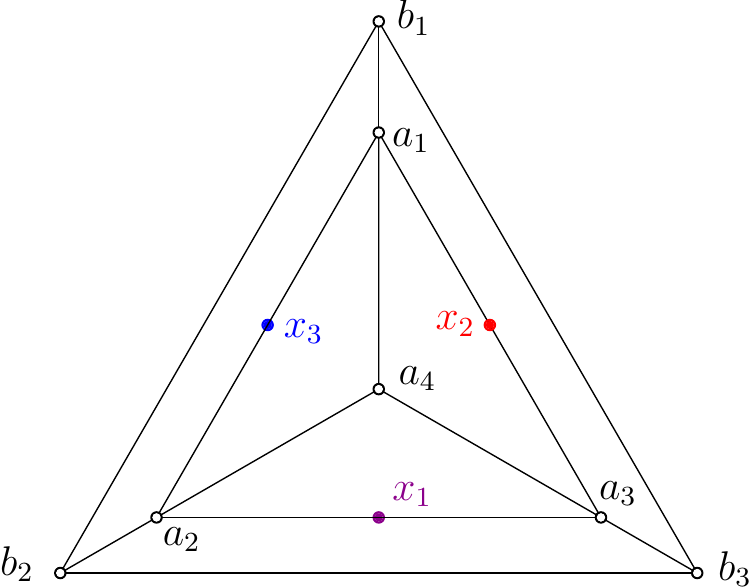}
	\caption{The graph $A$.}
	\label{fig:graphD}
\end{figure}

Consider any instance $(G,T)$ to {\sc MultiwayCut} with $|T|=3$.
We will use $n=|V(G)|$ and $n^2$ as a rough upper bound to $|E(G)|$.
We construct an edge-weighted graph $H=H(G,T)$ as follows:
\begin{itemize}
	\item[(i)] Construct $A$ and assign weight $n^5$ to its edges.
	\item[(ii)] Construct the graph $H'= G \cup A$, where the edges of $G$ have weight $1$.
	\item[(iii)] We identify each vertex of $T$ with a distinct vertex $x_i$ of $A$.
		That is, if $T=\{ t_1,t_2,t_3\}$,
		then, for each $i\in [3]$, identify $x_i$ and $t_i$.
\end{itemize}
This finishes the construction of $H$. 
See Figure~\ref{fig:graphH} for an example. 
Let $\pi$ be any rotation system for $H$ such that:

\begin{itemize}
	\item the restriction of $\pi$ to $A$ is the unique combinatorial embedding of $A$.
	\item for each $i\in [3]$, the edges $x_ia_{i-1}$ and $x_i a_{i+1}$ are consecutive in the 
		cyclic ordering $\pi_{x_i}$. That is, any edge of $H-E(A)$ incident to $x_i$
		is between $x_ia_{i+1}$ and $x_ia_{i-1}$ in the rotation system $\pi_{x_i}$.
\end{itemize} 

\begin{figure}
\centering
	\includegraphics[scale=.9]{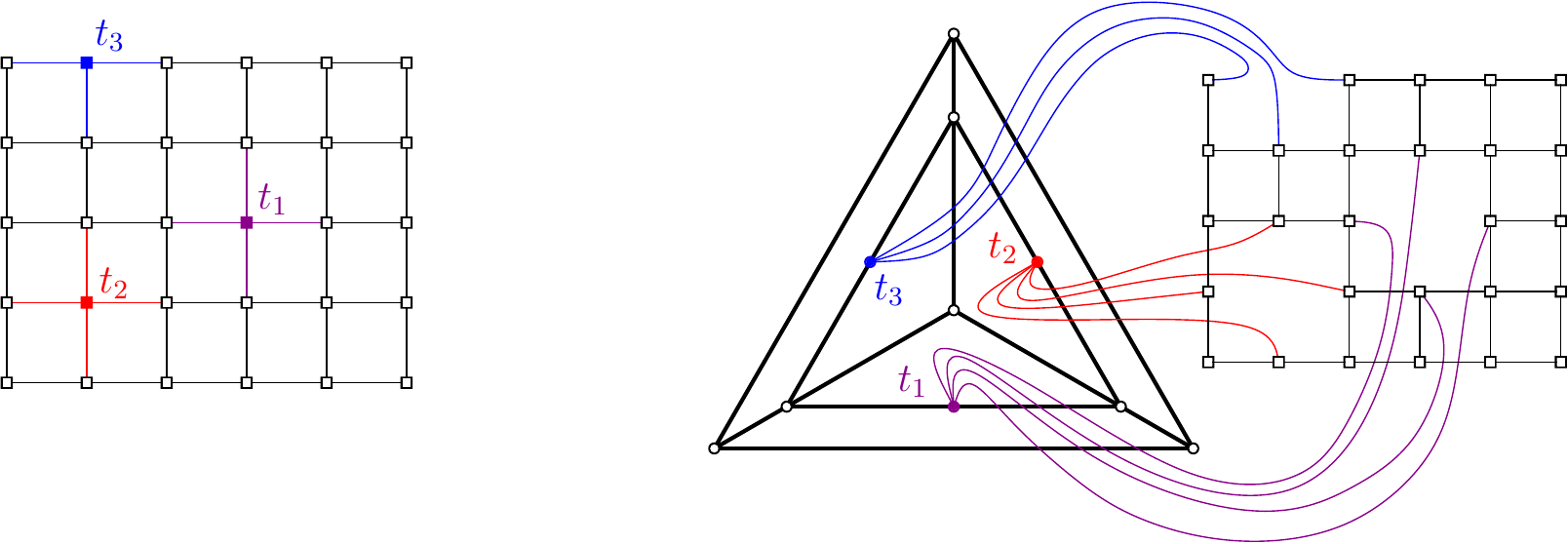}
	\caption{Left: a graph $G$ with vertices $T=\{t_1,t_2,t_3\}$ marked with filled-in squares.
		Right: the corresponding graph $H=H(G,T)$. Thicker edges have weight $n^5=30^5$ and the other edges have weight $1$.
		The rotation system of the drawing is a possible $\pi$.}
	\label{fig:graphH}
\end{figure}

In the next two lemmata
we obtain bounds relating $\CR(H)$ and $\CR(H,\pi)$ to $\MC(G,T)$.
Our bounds are not tight, but this does not affect our eventual results.

\begin{lemma}
\label{le:cro<}
	We have
	\[
		\CR(H,\pi) ~\le~ n^5\cdot \MC(G,T) + 3 n^4.
	\]
\end{lemma}
\begin{proof}
	Let $F$ be an optimal solution to {\sc MultiwayCut} for $(G,T)$. 
	Thus we have $|F|=\MC(G,T)$. For each $i\in [3]$, let $G_i$ be the connected
	component of $G-F$ that contains $x_i$. By the optimality of $F$, we have 
	$G=\left( \cup_{i\in [3]} G_i\right) + F$. 
	Indeed, if there would be another connected component, any edge of $F$ connecting that component
	to any other connected component could be removed from $F$ and obtain a better solution.
		
	We construct a drawing $\D$ of $H$ as follows.
	Firstly, take an embedding $A$ without any crossings; such embedding is shown in Figure~\ref{fig:graphD}.
	Then, for each $i\in [3]$, draw the component of $G_i$ inside the region limited by $x_i a_{i-1}a_4 a_{i+1}x_i$ 
	respecting the rotation system $\pi$ and with the minimum number of crossings. 
	Finally, draw each edge of $F$ optimally in the current drawing. 
	In such drawing, an edge connecting $G_i$ to $G_j$
	will be drawn crossing the edge $a_4 a_k$, where $k\not= i,j$.
	See Figure~\ref{fig:sketch} for a sketch. Let $\D$ be the resulting drawing.
	
	We now bound the number of crossings in the drawing $\D$. The restriction $\D(A)$ has no crossings by construction.
	For each $i\in [3]$, the restriction $\D(G_i)$ has at most $|V(G_i)|^4$ crossings
	because, as mentioned in Section~\ref{sec:prelim}, for any rotation system $\pi'$ of the complete graph $K_t$ we have $\CR(K_t,\pi') \le t^4$.
	Each single edge of $F$ can be drawn with $n^5+2n^2$ crossings.
	Indeed, if $v_iv_j\in F$ connects $G_i$ to $G_j$ and we denote by $k$ the element of $[3]\setminus \{i,j\}$,
	there is an arc from $v_i$ to any point on $\D(a_4 a_k)$
	that crosses at most $|E(G_i)|+|F|$ edges, and there is an arc connecting any point in $\D(a_4 a_k)$
	to $v_j$ with at most $|E(G_j)|+|F|$ crossings. The described drawing of $v_iv_j$ has, using a very rough estimate,
	\[
		\left( |E(G_i)|+|F|\right) + n^5+ \left( |E(G_j)|+|F| \right) ~\le~ n^5 + 2 |E(G)| ~\le~ n^5 + 2 n^2
	\]
	crossings. We conclude that
	\begin{align*}
		\CR(\D) ~&=~ \sum_{i\in[3]} |V(G_i)|^4 + \sum_{e\in F} (n^5 + 2n^2)\\
				~&\le~ n^4 + |F|\cdot n^5 + 2\cdot |F| \cdot n^2\\
				~&\le~ \MC(G,T) \cdot n^5 + 3 n^4 \qedhere
	\end{align*}	
\end{proof}
\begin{figure}
\centering
	\includegraphics[width=.45\textwidth]{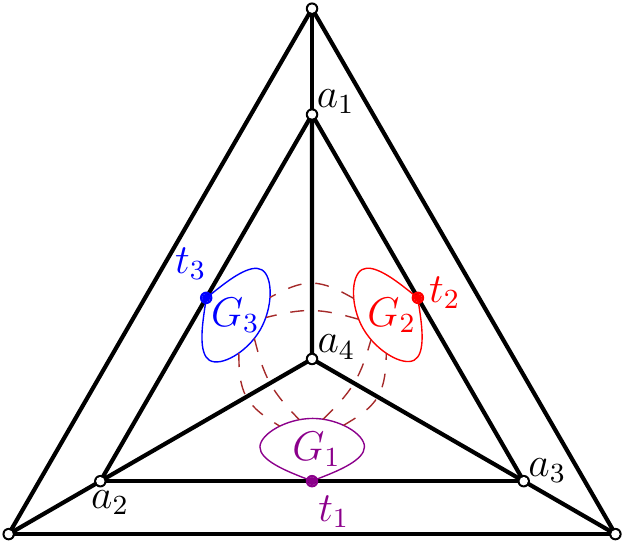}
	\caption{Sketch of the drawing in Lemma~\ref{le:cro<}. The dashed arcs represent edges of $F$.}
	\label{fig:sketch}
\end{figure}

The following result is independent of rotation systems.
\begin{lemma}
\label{le:cro>}
	From any drawing $\D$ of $H$ we can obtain in polynomial time
	a feasible solution $F$ to {\sc MultiwayCut}$(G,T)$ such that
	$|F|\le \CR(\D) /n^5$.
	In particular, 
	\[
		n^5\cdot \MC(G,T) ~\le~ \CR(H).
	\]
\end{lemma}
\begin{proof}
	Consider any drawing $\D$ of $H$.
	If $\CR(\D)$ is larger than $n^7$, we can just take $F=E(G)$ because it satisfies $|F|\le n^2\le \CR(\D)/n^5$.
	If $\CR(\D)$ is smaller than $n^7$, we proceed as follows.
	The restriction of $\D$ to $A$ is an embedding because
	each edge of $A$ has weight $n^5$. 
	For each $i\in [3]$, let $C_i$ denote the cycle $a_4a_{i-1}b_{i-1} b_{i+1} a_{i+1} a_4$.
	In the embedding $\D(A)$ the cycle $C_i$ separates $x_i=t_i$ from $x_j=t_j$, whenever $i\not= j$.

	Define the set of edges 
	\[
		F ~=~ \{ e\in E(G)\mid \mbox{$\D(e)$ intersects $\D(C_1)$, $\D(C_2)$ or $\D(C_3)$}\}.
	\]
	Note that $F$ can be computed in polynomial time from $\D$.
	Since each edge of $F$ crosses (at least once) some edge of $A$, we have 
	\[
		\CR(\D) ~\ge~ n^5 \cdot |F|.
	\]
	Furthermore, $F$ is a feasible solution to {\sc MultiwayCut} for $(G,T)$
	because, for each path $P$ in $G$ that connects $t_i$ to $t_j$, $i\not= j$, the drawing $\D(P)$ has to cross the cycle $\D(C_i)$ and thus $P$ has
	an edge in $F$.
	
	The bound $n^5\cdot \MC(G,T) ~\le~ \CR(H)$ is obtained by considering an optimal drawing $\D^*$ of $H$.
	Such drawing $\D^*$ gives a feasible solution $F$ that satisfies
	\[
		n^5\cdot \MC(G,T) ~\le~ n^5\cdot |F| ~\le~ \CR(\D^*) ~=~ \CR(H).
	\]
	The result follows.	
\end{proof}

We next explain how to construct a cubic graph $\tilde H=\tilde H(G,T)$ such that 
\[
	n^5\cdot \MC(G,T) \le \CR(\tilde H) \le n^5\cdot \MC(G,T) + 3 n^4.
\]
The idea is a straightforward adaptation of the technique used by Pelsmajer et~al.~\cite{PSS}; 
we include the details for the sake of completeness.

In a first step, we construct the unweighted graph $H'=\phi(H)$ and the rotation system
$\pi'=\phi_\pi(H(G,T),\pi)$ in $H'$, as described in Section~\ref{sec:prelim}.
It holds that $\CR(H')=\CR(H)$ and $\CR(H',\pi')=\CR(H,\pi)$.

In a second step, we replace each vertex $v\in H'$ by a cubic grid $\mathcal{C}_v$ 
of width $\deg_{H'}(v)$ and height $4n^7$; see Figure~\ref{fig:grid}.
(The cubic grid of width $d$ and height $h$ is obtained from a regular, rectangular grid of width $2d$ and height $h$
where the vertical edge connecting $(i,j)$ to $(i,j+1)$ is removed whenever $i+j$ is odd.) 
If $e^v_1,\dots,e^v_{\deg_{H'}(v)}$ are the edges incident to $v$ in $H'$ ordered as in the cyclic ordering $\pi'_v$,
then we attach the edges $e^v_1,\dots,e^v_{\deg_{H'}(v)}$ to the degree-two consecutive vertices of the cubic grid $\mathcal{C}_v$
that are on the higher row. 
Finally, we make the graph cubic by removing vertices of degree $1$ and contracting some edges incident to vertices
of degree $2$.  This finishes the construction of $\tilde H=\tilde H(G,T)$.
Note that the construction of $\tilde H$ can be made in polynomial time because the weight of each edge of $H$ is bounded by $n^5$.

\begin{figure}
\centering
	\includegraphics[width=.6\textwidth]{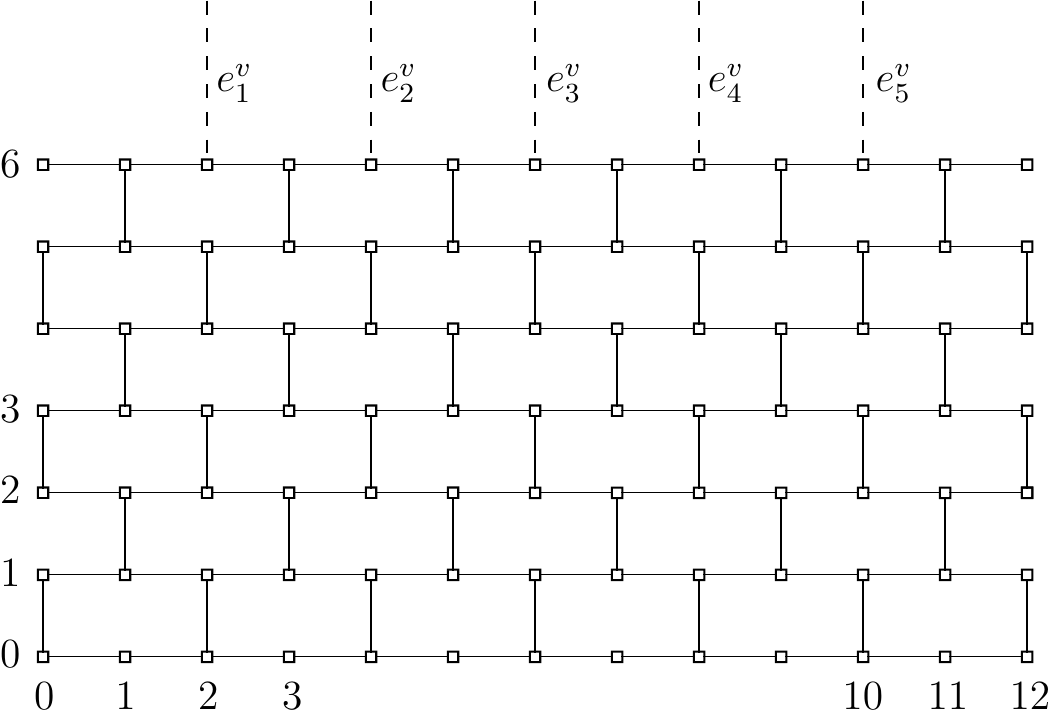}
	\caption{The solid edges form a cubic grid of width $6$ and height $6$. 
		The dashed edges show how the edges $e^v_1,e^v_2,\dots$ get attached to $\mathcal{C}_v$.}
	\label{fig:grid}
\end{figure}

\begin{lemma}
\label{le:tilde}
	We have 
	\[
		n^5\cdot \MC(G,T) ~\le~ \CR(\tilde H) ~\le~ n^5\cdot \MC(G,T) + 3 n^4.
	\]
	Furthermore, from any drawing $\tilde \D$ of $\tilde H$ we can obtain in polynomial time
	a feasible solution $F$ to {\sc MultiwayCut}$(G,T)$ such that $|F|\le \CR(\tilde \D) /n^5$.
\end{lemma}
\begin{proof}
	It is clear that any drawing $\D'$ of $H'$ with rotation system $\pi'$ can be 
	converted into a drawing of $\tilde H$ by a local replacement around $\D'(v)$, for each $v\in V(H')$, 
	without introducing additional crossings. Therefore 
	\[
	\CR(\tilde H) ~\le~ \CR(H',\pi') ~=~ \CR(H,\pi) \le n^5\cdot \MC(G,T) + 3 n^4
	\]
	because of Lemma~\ref{le:cro<}.
	
	To see the other inequality, consider an optimal drawing $\tilde\D$ of $\tilde H$.
	The first part of the proof implies that $\CR(\tilde\D)<4n^7$.
	Therefore in each cubic grid $\mathcal{C}_v$, $v\in V(H')$,
	there is at least one horizontal row, let's call it $R_v$, that does not participate in any crossing of $\tilde \D$.
	For each vertex $v\in V(H')$, there are $\deg_{H'}(v)$ vertex-disjoint paths $P^v_i$ in $\mathcal{C}_v$, $i=1\dots \deg_{H'}(v)$,
	connecting the endvertex of $e^v_i$ to a vertex of $R_v$. We contract the row $R_v$ to a point
	and remove from the drawing all the edges of $\mathcal{C}_v$, but those participating in the paths $P^v_1,\dots, P^v_{\deg_{H'}(v)}$.
	Repeating this for each vertex $v\in V(H')$ we obtain a drawing $\D'$ of a subdivision of $H'$ 
	with at most $\CR(\tilde \D)$ crossings.
	This implies that $\CR(H)=\CR(H') \le \CR(\tilde D) = \CR(\tilde H)$. 
	By Lemma~\ref{le:cro>} it follows that
	\[
		\CR(\tilde H) ~\ge~ \CR(H) ~\ge~ n^5\cdot \MC(G,T).
	\]
	
	To obtain from a drawing $\tilde \D$ of $\tilde H$ a feasible solution $F$ with $|F|\le \CR(\tilde \D) /n^5$,
	we proceed as follows. If $\CR(\tilde \D)\ge n^7$, we just return $F=E(G)$. 
	Otherwise we construct the drawing $\D'$ of $H'$ from $\tilde \D$, as described above. 
	As discussed in Section~\ref{sec:prelim}, from
	the drawing $\D'$ of $H'=\phi(H)$ we can obtain a drawing $\D$ of $H$
	with $\CR(\D)\le \CR(\D')$.
	Finally, from the drawing $\D$ of $H$ we can use Lemma~\ref{le:cro>} to extract a feasible solution 
	$F$ to {\sc MultiwayCut}$(G,T)$ such that 
	\[
		|F|~\le~ \CR(\D) /n^5 ~\le~ \CR(\D')/n^5 ~\le~ \CR(\tilde\D).
	\]	
	Since all the steps can be carried out in polynomial time, the result follows.
\end{proof}

\medskip

\begin{theorem}
\label{thm:hard1} 
	There is a constant $c_0>1$ such that, if $P\not=$NP,
	there is no $c_0$-approximation algorithm for {\sc CrossingNumber}, even when restricted to cubic graphs.
\end{theorem}
\begin{proof}
	Let $c_M>1$ be a constant such that it is NP-hard to compute a $c_M$-approximation to {\sc MultiwayCut} when $|T|=3$.
	(See the discussion in Section~\ref{sec:prelim}.)
	Take $c_0=c_M-\eps$ for an arbitrary constant $\eps$ with $0<\eps<c_M-1$. 
	We will see that it is NP-hard finding a $c_0$-approximation to {\sc CrossingNumber} in cubic graphs. 

	Assume, for the sake of contradiction, that there is an $c_0$-approximation algorithm for 
	{\sc CrossingNumber} in cubic graphs. 
	We can then obtain a $c_M$-approximation
	to {\sc MultiwayCut}$(G,T)$ in polynomial time, as follows. 
	Let $n=|V(G)|$.
	If $n$ is smaller than $3c_0/\eps$, which is a constant, we run any brute force algorithm.
	Otherwise, we construct in polynomial time the cubic graph $\tilde H=\tilde H(G,T)$, as described above, 
	use the $c_0$-approximation algorithm to compute a drawing $\tilde \D$ of $\tilde H$ with 
	$\CR(\tilde \D)\le c_0\cdot \CR(\tilde H)$, use Lemma~\ref{le:tilde} to find a feasible solution 
	$F$ to {\sc MultiwayCut}$(G,T)$, and return $F$.
	We next argue that this algorithm is a $c_M$-approximation algorithm for {\sc MultiwayCut}.
	
	Because of Lemma~\ref{le:tilde}, $F$ is a feasible solution with $|F|\le \CR(\tilde \D) /n^5$.
	On the other hand, $\CR(\tilde \D) ~\le~ c_0\cdot \CR(\tilde H)$ because $\D$ is a $c_0$-approximation to $\CR(\tilde H)$. 
	Using Lemma~\ref{le:tilde} we obtain
	\begin{align*}
		|F|~&\le~ \frac{\CR(\tilde \D)}{n^5} \\
			&\le~ c_0 \cdot \frac{\CR(\tilde H)}{n^5}\\
			&\le~ c_0 \cdot \frac{n^5\cdot \MC(G,T) + 3 n^4}{n^5}\\
			&\le~ c_0 \cdot \MC(G,T) + 3c_0/n\\
			&\le~ c_0 \cdot \MC(G,T)+ \eps \\
			&\le~ (c_0+\eps) \cdot \MC(G,T)	\\
			&=~ c_M \cdot \MC(G,T).		
	\end{align*}
	Thus, returning $F$ we obtain a $c_M$-approximation to $\MC(G,T)$,
	which is not possible unless P$\not=$NP.
\end{proof}

Bokal et al.~\cite{BFM} introduced the concept of minor crossing number.
Hlin\v{e}n\'y~\cite{Hli} noted that for cubic graphs the crossing number and the minor crossing number have the same value.
We thus obtain the following. 

\begin{corollary}
\label{cor:hard1} 
	There is a constant $c_0>1$ such that, if $P\not=$NP,
	there is no $c_0$-approximation algorithm for the minor crossing number.
\end{corollary}

\section{Conclusions}
Since there are constant-factor approximation algorithms for {\sc MultiwayCut}, 
a more careful reduction from {\sc MultiwayCut} will not bring us beyond hardness of constant-factor approximations. 
Nevertheless, it seems hard to believe that there is an $O(1)$-approximation algorithm for {\sc CrossingNumber}.
As mentioned in the introduction, the currently best approximation factor is roughly $O(n^{9/10})$.

A natural approach to improve the inapproximability result would be to reduce from a problem that
is known to be harder. 
The problems \emph{$0$-extension} and \emph{MetricLabeling} are generalizations of {\sc MultiwayCut}
and stronger inapproximability results are known~\cite{cn-07,kkmr-09,mnrs-08}. 
However, we have not been able to obtain fruitful reductions from those problems to {\sc CrossingNumber}.

It remains a tantalizing open problem whether {\sc CrossingNumber} can be solved in polynomial
time for graphs with bounded treewidth. 
An obstacle is that we do not know whether {\sc LinearArrangement} is NP-hard 
for graphs of bounded treewidth. If that would be the case, then
the reduction of Garey and Johnson~\cite{GJ} increases the treewidth by a constant. 
On the other hand, {\sc MultiwayCut} is solvable in polynomial-time for graphs of bounded treewidth:
Chopra and Rao~\cite{cr-91} discuss treewidth 2 and Dahlhaus et al.~\cite{djpsy-94} note that it works for any bounded treewidth.
Thus, the approach of this paper cannot lead to an NP-hardness proof of {\sc CrossingNumber} for graphs of bounded treewidth.

\section*{Acknowledgments} 
I am grateful to Bojan Mohar for discussions and to Petr Hlin\v{e}n\'y for pointing out the extension to cubic graphs.

\bibliographystyle{abbrv}
\bibliography{bibliography}
\end{document}